\documentclass{scrartcl}

\usepackage{Prelude}
\usepackage{dirtytalk}

\DeclareMathOperator{\indeg}{InDeg}
\DeclareMathOperator{\outdeg}{OutDeg}
\DeclareMathOperator{\innbr}{InNbr}
\DeclareMathOperator{\outnbr}{OutNbr}

\newcommand{\KVALUE}{\lceil 2m/\eps \rceil}

\newif\ifnames

\namestrue

\title{Efficient Catalytic Graph Algorithms\footnote{%
This paper subsumes a manuscript by the first author which showed how to simulate random walks on acyclic graphs~\cite{ECCCWalks}.%
}}
\ifnames
\author{James Cook \\[-4pt]\large\texttt{falsifian@falsifian.org} \and Edward Pyne \\[-4pt]\large\texttt{epyne@mit.edu}}
\else 
\author{Anonymous authors}
\fi

\begin{document}
\begin{titlepage}
\maketitle

\begin{abstract}
We give fast, simple, and implementable catalytic logspace algorithms for two fundamental graph problems. 

First, a randomized catalytic algorithm for $s\to t$ connectivity running in $\widetilde{O}(nm)$ time, and a deterministic catalytic algorithm for the same running in $\widetilde{O}(n^3 m)$ time. The former algorithm is the first algorithmic use of randomization in $\mathsf{CL}$. The algorithm uses one register per vertex and repeatedly ``pushes'' values along the edges in the graph.

Second, a deterministic catalytic algorithm for simulating random walks which in $\widetilde{O}( m T^2 / \varepsilon )$ time estimates the probability a $T$-step random walk ends at a given vertex within $\varepsilon$ additive error.
The algorithm uses one register for each vertex and increments it at each visit to ensure repeated visits follow different outgoing edges.

Prior catalytic algorithms for both problems did not have explicit runtime bounds beyond being polynomial in $n$.
\end{abstract}

\vfill

{
\centering
\STConnFigure{T}{F}
\nobreak\quad\nobreak\( \Rightarrow \)\nobreak\quad\nobreak
\STConnFigure{T}{T}
}

\smallskip

\noindent
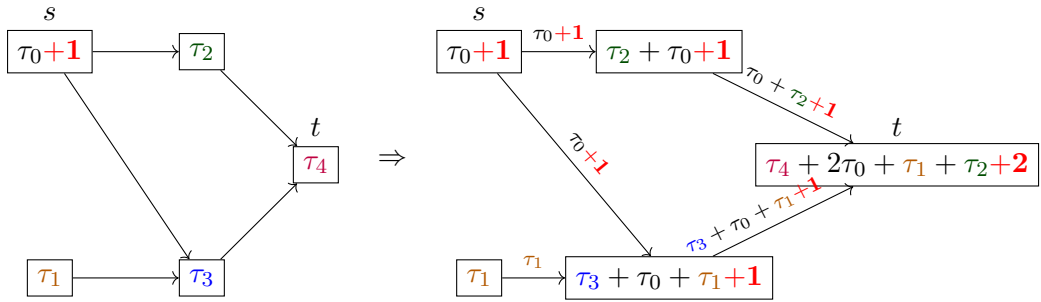
Detecting \( s \to t \) paths by adding values along edges: a change (\( +1 \)) to \( s \) reaches \( t \).

\bigskip

{
\centering
% Generated by ./make_figures. Any changes to this file will be overwritten.
\begin{tikzpicture}
\path
  +(0,0.6) node (00) [active configuration] {\( \boldsymbol\uparrow \)}
  +(1.3,1.2) node (10) [active configuration] {\( \boldsymbol\downarrow \)}
  +(1.3,0) node (11) [inactive configuration] {\( \uparrow \)}
  +(2.6,1.8) node (20) [inactive configuration] {\( \downarrow \)}
  +(2.6,0.6) node (21) [active configuration] {\( \boldsymbol\uparrow \)}
  +(2.6,-0.6) node (22) [inactive configuration] {\( \downarrow \)}
  +(3.9,1.2) node (30) [active configuration] { }
  +(3.9,0) node (31) [inactive configuration] { }
  (00) edge [active edge] (10)
  (00) edge [inactive edge] (11)
  (10) edge [inactive edge] (20)
  (10) edge [active edge] (21)
  (11) edge [inactive edge] (21)
  (11) edge [inactive edge] (22)
  (20) edge [inactive edge] (30)
  (20) edge [inactive edge] (31)
  (21) edge [active edge] (30)
  (21) edge [inactive edge] (31)
  (22) edge [inactive edge] (30)
  (22) edge [inactive edge] (31)
  ;
\end{tikzpicture}

% Generated by ./make_figures. Any changes to this file will be overwritten.
\begin{tikzpicture}
\path
  +(0,0.6) node (00) [active configuration] {\( \boldsymbol\downarrow \)}
  +(1.3,1.2) node (10) [inactive configuration] {\( \uparrow \)}
  +(1.3,0) node (11) [active configuration] {\( \boldsymbol\uparrow \)}
  +(2.6,1.8) node (20) [inactive configuration] {\( \downarrow \)}
  +(2.6,0.6) node (21) [active configuration] {\( \boldsymbol\downarrow \)}
  +(2.6,-0.6) node (22) [inactive configuration] {\( \downarrow \)}
  +(3.9,1.2) node (30) [inactive configuration] { }
  +(3.9,0) node (31) [active configuration] { }
  (00) edge [inactive edge] (10)
  (00) edge [active edge] (11)
  (10) edge [inactive edge] (20)
  (10) edge [inactive edge] (21)
  (11) edge [active edge] (21)
  (11) edge [inactive edge] (22)
  (20) edge [inactive edge] (30)
  (20) edge [inactive edge] (31)
  (21) edge [inactive edge] (30)
  (21) edge [active edge] (31)
  (22) edge [inactive edge] (30)
  (22) edge [inactive edge] (31)
  ;
\end{tikzpicture}

% Generated by ./make_figures. Any changes to this file will be overwritten.
\begin{tikzpicture}
\path
  +(0,0.6) node (00) [active configuration] {\( \boldsymbol\uparrow \)}
  +(1.3,1.2) node (10) [active configuration] {\( \boldsymbol\uparrow \)}
  +(1.3,0) node (11) [inactive configuration] {\( \downarrow \)}
  +(2.6,1.8) node (20) [active configuration] {\( \boldsymbol\downarrow \)}
  +(2.6,0.6) node (21) [inactive configuration] {\( \uparrow \)}
  +(2.6,-0.6) node (22) [inactive configuration] {\( \downarrow \)}
  +(3.9,1.2) node (30) [inactive configuration] { }
  +(3.9,0) node (31) [active configuration] { }
  (00) edge [active edge] (10)
  (00) edge [inactive edge] (11)
  (10) edge [active edge] (20)
  (10) edge [inactive edge] (21)
  (11) edge [inactive edge] (21)
  (11) edge [inactive edge] (22)
  (20) edge [inactive edge] (30)
  (20) edge [active edge] (31)
  (21) edge [inactive edge] (30)
  (21) edge [inactive edge] (31)
  (22) edge [inactive edge] (30)
  (22) edge [inactive edge] (31)
  ;
\end{tikzpicture}

}

\smallskip

\noindent
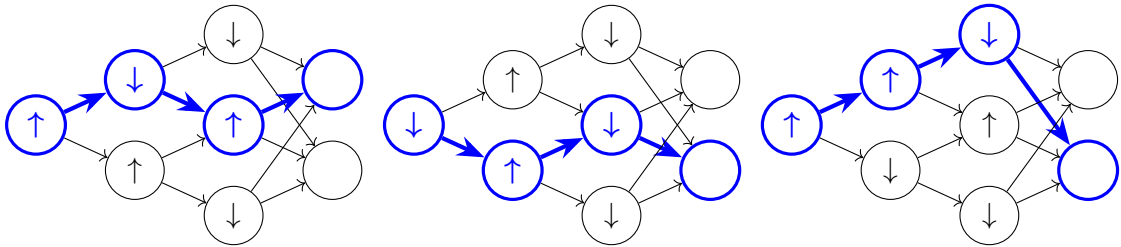
Simulating a random walk, using one bit per node to alternate between random choices.

\vfill

\bigskip

\thispagestyle{empty}
\end{titlepage}

\section{Introduction}

\subsection{Catalytic space}

In the catalytic space model, an algorithm has two tapes to use as memory: a smaller ordinary tape, and a larger ``catalytic'' tape which starts out filled with arbitrary data.
The algorithm may freely write to and read from both tapes, but when it finishes, the catalytic tape's original content must be restored.
Often the working tape has size \( O( \log n ) \) and the catalytic tape has size \( \Poly( n ) \), defining the decision class catalytic logspace or \( \CL \).
Buhrman, Cleve, Kouck\'y, Loff and Speelman~\cite{Catalytic} initiated the study of this model with their surprising result that evaluation of logarithmic depth threshold circuits was possible in $\CL$, from which it follows for example that $s\ra t$ connectivity and estimation of random walks (the complete problems for nondeterministic logspace ($\NL$) and randomized logspace ($\BPL$) respectively) are in \( \CL \).

Broadly speaking, two ways to take advantage of catalytic space have been explored previously: compression-based techniques are able to make use of incompressible catalytic tapes (as randomness, for example); and algebra-based techniques treat the tape as an array of registers which they use to execute a ``straight-line program'', a pre-determined sequence of mathematical operations.
The original \( \TC^1 \subseteq \CL \) result is an example of the algebraic approach.
The compression-based approach began with a ``compress-or-random'' argument that \( \BPL \subseteq \CL \) (Mertz~\cite{MertzSurvey} gives a sketch of it). Several subsequent works~\cite{OpenDistinguisher,BplSmallCl,CBPL,CNL} applied this technique further, with Cook, Li, Mertz and Pyne~\cite{CBPL} using it to derandomize randomized $\CL$ itself. All of these results in both branches have a \say{complexity} flavor --- in particular, the runtime for the relevant problem is some large polynomial.\footnote{Cook~\cite{CookBlog2021} implemented a register program for $s\ra t$ connectivity with an estimated runtime of $\Theta( n^{8} )$, and catalytic tape size $\Theta(n^3 \log n)$.}

We give new techniques for deciding graph connectivity and estimating random walks. Our algorithms are very simple and implementable, and allow a precise analysis of the catalytic space consumption and runtime.\footnote{For the runtime bounds, we assume we have RAM access to the catalytic tape and oracle access to the graph. This does not affect the structure of the class (as we can simulate a catalytic RAM with a standard catalytic machine with a polynomial slowdown).} We hope this will initiate the study of $\CL$ from an algorithmic perspective.

\subsection{Our Results: Connectivity}
We first state our results for connectivity, beginning with our deterministic algorithm:
\begin{restatable}{theorem}{mainNL}\label{thm:mainNL}
    There is a catalytic algorithm that, given a simple directed graph with $n$ vertices and $m$ edges together with vertices $s,t$, decides whether there exists a path from $s$ to $t$.
    The algorithm uses $O(\log n )$ workspace and $\tO(n^2)$ catalytic space, and runs in time $\tO(n^3m)$.
\end{restatable}

Next, we give a randomized algorithm that improves the runtime to $\tO(nm)$ and space usage to $\tO(n)$, only a factor of $n$ in runtime from the linear-workspace BFS algorithm:
\begin{restatable}{theorem}{mainNLrand}\label{thm:mainNLrand}
    There is a randomized catalytic algorithm\footnote{The algorithm always resets the catalytic tape no matter the random coins. If there is no $s\ra t$ path, the algorithm always returns no, and if there is an $s\ra t$ path returns yes with probability at least $1/2$.}
    that, given a simple directed graph with $n$ vertices and $m\ge n$ edges together with vertices $s,t$, decides whether there exists a path from $s$ to $t$.
    The algorithm uses $O(\log n)$ workspace and $\tO(n)$ catalytic space, and runs in time $\tO(n m)$.
\end{restatable}
As far as we know, this result is the first \textit{use} for randomness in catalytic computing --- rather than using it to place a new problem in the class (now provably impossible~\cite{CBPL}), we use it to give an algorithmic speedup. 
\begin{remark}
    Under conjectured time-space tradeoffs for connectivity, the catalytic space usage of \cref{thm:mainNLrand} is optimal up to polylogarithmic factors, even for an algorithm running in arbitrarily large polynomial time.\footnote{It can be shown that a randomized catalytic logspace algorithm for $s\ra t$ connectivity using $n^{1-\eps}$ catalytic space would imply a randomized $n^{1-\eps}$ space, polynomial-time algorithm for $s\ra t$ connectivity, which has been open for decades~\cite{Wig92} and has been conjectured not to exist.}
\end{remark}

In fact, we can obtain an additional desirable property. The practical motivation for $\CL$ is to borrow temporarily unused space to perform useful computation. Unfortunately, if any part of the borrowed section of memory is needed during the computation of the catalytic machine, the original owner may need to wait for the entire catalytic computation to finish. As such, existing catalytic algorithms do not seem to permit sharing a single section of memory without huge latency.\footnote{Alongside the large runtime of prior algorithms, this was mentioned as the primary issue for practical catalytic computing by Cook~\cite{CookBlog2021}.} To rectify this problem, we define a notion of catalytic algorithms that must permit fast query access to the \textit{original} memory configuration at all times:
\begin{definition}
    A catalytic algorithm $A$ is \emphdef{locally revertible in time $t$} if at any point during the execution of $A$ with initial tape $\tau$, the algorithm can be paused and queried on an index $i$, and will return $\tau_i$ in time $t$ (and then continue its execution).
\end{definition}

We show that (at the cost of polynomially more catalytic space) our randomized connectivity algorithm can be made locally revertible for $t=\polylog(n)$:
\begin{restatable}{theorem}{mainNLrandrevert}\label{thm:mainNLrandrevert}
    There is a randomized catalytic algorithm
    that, given a simple
    directed graph with $n$ vertices and $m\ge n$ edges together with vertices $s,t$, decides whether there exists a path from $s$ to $t$.
    The algorithm uses $O(\log n)$ workspace and $\tO(n^3)$ catalytic space, and runs in time $\tO(nm)$. Moreover, the algorithm is locally revertible in time $\polylog(n)$.
\end{restatable}
As far as we are aware, no previous catalytic algorithm is locally revertible for any time bound smaller than its runtime. This new property strengthens the motivation for catalytic space: rather than borrowing an unused hard drive, the catalytic algorithm only needs to borrow the ability to \textit{write} to that hard drive, since at all times the original data will be accessible for reading with small latency.

\subsection{Our Results: Random Walks}
Next, we give an efficient catalytic simulation of random walks:
\begin{restatable}[random walk on a general graph]{theorem}{ThmRandomGeneral}\label{thm:randomGeneral}
    There is a catalytic algorithm that, given a graph with \( n \) vertices and \( m \ge n \) edges, together with vertices $s,t$ and parameters $T \in \N, \eps>0$, returns $\rho$ such that
    \[
    \left|\rho - \Pr[\text{$T$-step random walk from $s$ ends at $t$}]\right|\le \eps \ldotp
    \]
    The algorithm runs in time $\tO( mT^2/\eps )$ and uses $O(\log(mT/\eps))$ workspace and \( \tO( nT \cdot \log( m / \eps ) ) \) catalytic space.
\end{restatable}
Prior algorithms for random walks are primarily based on logspace derandomization techniques that are not practically efficient. For estimating sub-polynomial length walks to inverse sub-polynomial error, our algorithm runs in almost linear time $m^{1+o(1)}$. 

If \( G \) is guaranteed to be acyclic, the algorithm can be made to use less time and space by skipping a transformation that removes cycles (see \cref{thm:randomDag}).
Alternatively, na\"ively applying the algorithm for acyclic graphs on a graph with cycles produces an interesting result: the algorithm is no longer catalytic, but produces a walk with visit counts matching a true random walk's stationary distribution, after a number of steps which depends on the mixing time.
For details, see \cref{thm:stationary}.

\subsection{Our Technique: Connectivity}
We will give a catalytic algorithm that, given a graph $G=(V,E)$ on $n$ vertices, determines if there exists a path from $s$ to $t$. We virtually lift $G$ to a layered graph on $(n+1)$ layers with vertex set $\{0,\ldots,n\}\times V$, where we place an edge from $(i,u)$ to $(i+1,v)$ if $(u,v)\in E$.

Next, for every $v\in V$ and timestep $i\in \{0,\ldots,n\}$, allocate an $\ell=\tO(n)$ bit register on the catalytic tape, which we denote $R_{(i,v)}$, and interpret the register as a number in $\Z/2^\ell\Z$. Let the initial value of the register be $\tau_{(i,v)}$.

We define an \emphdef{edge push} as, for an edge $((i,u),(i+1,v))$ in the lifted graph, setting
\[
R_{(i+1,v)}\la R_{(i+1,v)}+R_{(i,u)}.
\]
For each layer $i=0,\ldots,n-1$ in sequence, we perform an edge push for every edge in the layer. 

Let $\alpha_{(n,t)}$ be the final value of $R_{(n,t)}$ after pushing along every edge. Note that we can easily revert the catalytic tape by performing a \emphdef{reverse edge push} on every edge in every layer, where we subtract the register instead of adding. Now consider incrementing $R_{(0,s)}$ (i.e.\ the register of the start vertex in the first layer) by $1$ and performing the same sequence of edge pushes; let the new final value of $R_{(n,t)}$ be $\alpha'_{(n,t)}$. We show via an inductive argument that $\alpha_{(n,t)}'-\alpha_{(n,t)}$ is exactly the number of length-$n$ paths from $s$ to $t$, modulo $2^\ell$ (and by adding a self-loop to $t$ we can assume that if an $s\ra t$ path exists, there exists an $s\ra t$ path of length $n$). By repeatedly executing the edge push sequence and comparing $\alpha_{(n,t)}$ to $\alpha'_{(n,t)}$ bit by bit, we determine whether their difference is nonzero --- equivalently, whether there exists an $s\ra t$ path.

\begin{figure}
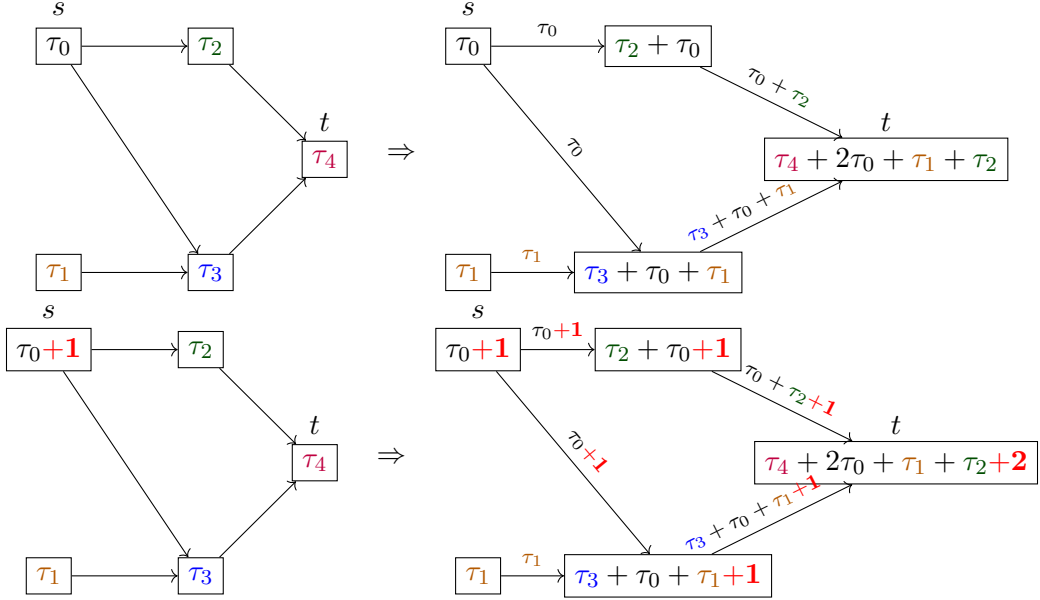

\centering
\STConnFigure{F}{F}
\nobreak\quad\nobreak\( \Rightarrow \)\nobreak\quad\nobreak
\STConnFigure{F}{T}

\STConnFigure{T}{F}
\nobreak\quad\nobreak\( \Rightarrow \)\nobreak\quad\nobreak
\STConnFigure{T}{T}
\caption{Pushing values along edges to detect whether a \( + 1 \) is propagated from \( s \) to \( t \).}
\end{figure}

\paragraph{Improving the runtime with randomness}
Unfortunately, since the number of paths from $s$ to $t$ can be exponential in $n$ and we count the number of paths mod the register size, our deterministic algorithm must take the register size $\ell$ to be $\Omega(n)$, so pushing a single edge takes linear time. Moreover, we must compute the entire sequence of pushes $\Omega(\ell / \log n)=\widetilde{\Omega}(n)$ times to compare $\alpha_{(n,t)}$ to $\alpha'_{(n,t)}$. To avoid this slowdown, we use randomness. Instead of working mod $2^{ \ell }$ for $\ell=\Omega(n)$, we pick a random small modulus $q$ and work mod $q$. If the number of paths is nonzero, with reasonable probability we will have
\[
\alpha_{(n,t)}-\alpha'_{(n,t)} \not\equiv 0 \pmod q
\]
from which the algorithm can infer that there is a path from \( s \) to \( t \).
In this way, we reduce the size of each register to $O(\log n)$ bits.

If \( q \) is not a power of two, not all strings of length $\ell$ will correspond to values mod $q$, so we must ensure that each register is \say{valid} before running the algorithm. This problem was  solved before by Buhrman et.\ al.~\cite{Catalytic}, but their approach is not fast enough for us.
Instead, we use a larger modulus \( dq \) chosen so that \( 2^{ \ell } - dq \) is small, and then draw a random shift $\beta \in [2^{ \ell }]$ (which we record on the worktape) and add \( \beta \) to every register.
With high probability over $\beta$, this ensures every register contains a value less than $dq$ (and otherwise we abort).

\paragraph{Local revertibility}
Next, we discuss how to achieve local revertibility. After performing all edge pushes onto a register $R_{a=(i+1,v)}$, the current value of this register is exactly its original value $\tau_{a}$ plus 
\[
\sum_{u:(u,v)\in E}R_{(i-1,u)}
\]
where $R_{(i-1,u)}$ is the current value of the register in the previous layer. Thus, if we receive a query for the initial value $\tau_a$ of this register, we can iterate over the in-neighbors of $v$, subtract off the register values of the previous layer, and thus return $\tau_a$. The time for this operation is essentially the product of the in-degree of $v$ and register size. Since we have already reduced the register size to $O(\log n)$, it suffices to ensure our graph has bounded in-degree, which we show we can do with a standard transformation (\cref{lem:reducedegree}).

\paragraph{Decreasing the number of registers}
Finally, as written we use $n^2$ registers, one per vertex and timestep. In fact, we can simply use two sets of registers, one for odd and one for even timesteps, and alternate. This saves a factor of $n$ in space, but breaks local revertibility. The version of our algorithm with this modification is \cref{thm:stnonzero}.

\subsection{Our Technique: Random Walks}

In this introduction, for simplicity, we will assume the graph \( G \) is acyclic and 2-out\-regular, with each vertex having a $0$ edge and $1$ edge.
We will determine the probability that a random walk from $s$ ends at $t$ with additive error at most $\eps$.

Allocate one bit of the catalytic tape for each vertex. For vertex $v$, denote this register $R_v$.
Run \( K=\KVALUE \) walks from $s$ as follows. At vertex $v$, we take the edge labeled with the current value of $R_v$, then set $R_v \gets 1 - R_v$.
In this way, if we examine walks reaching $v$ over the course of the \( K \) walks, the next edges taken are  \( 0, 1, 0, 1, \dotsc \) or \( 1, 0, 1, 0, \dotsc \).
The top row of \cref{fig:probsAndCounts} shows an example of this process, with the bits of the catalytic tape drawn directly on the corresponding vertices as \( \uparrow \) or \( \downarrow \).

Now we can argue that the number of visits to each vertex approximately equals the expected number of visits if we had done a truly random walk.
A common way to prove this kind of result is with a ``local consistency check'' --- see for example Nisan's Lemma~2.6~\cite{ROMultiple,CH22}.

The general idea is that if every vertex is given a pseudorandom bit of \( 0 \) approximately as many times as it is given a \( 1 \), then each vertex is visited approximately the right number of times.
Our version of this argument appears in the proof of \cref{lem:accurate}.
By a careful analysis we show that the number of visits to each vertex is within \( 2m \) of its expected value, regardless of the number of simulations. Thus after $K$ simulations, we obtain additive error at most $2m/K \le \eps$.

Finally, we must restore the catalytic tape.
This is done by running the ``reverse'' of our \( K \) simulations.
We do not literally walk in reverse; rather, we walk forward as before, but slightly change the way we supply its random bits, so that each ``reverse'' walk exactly undoes the effect of one normal walk.
For details, see \cref{alg:walkOnce} and \cref{cor:isCatalytic}.

\begin{figure}
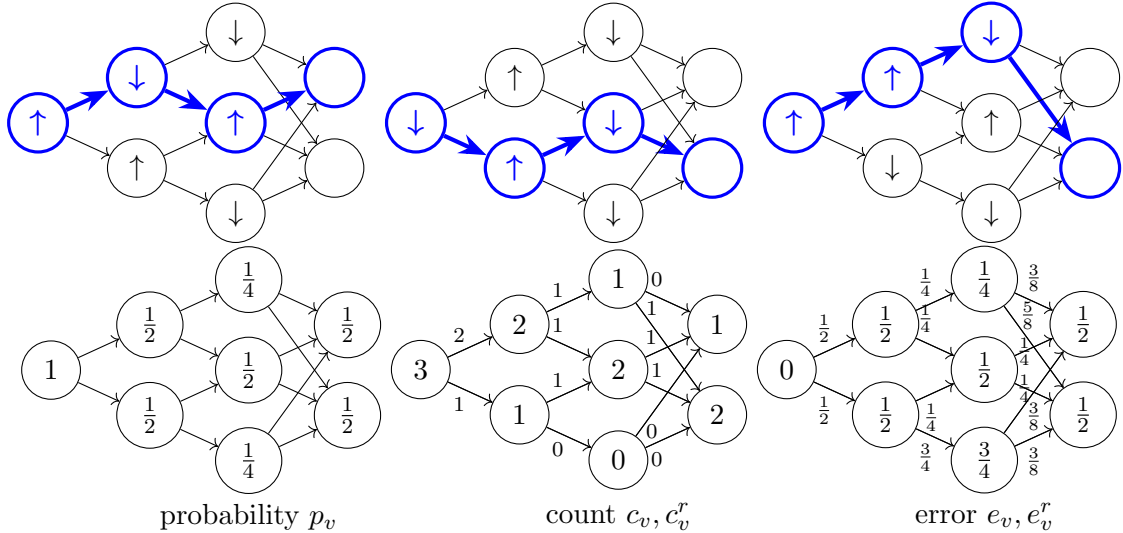

\centering
\input{gen_step0}
\input{gen_step1}
\input{gen_step2}

\begin{tikzpicture}
\input{gen_probs}
\path (22) +(0,-0.75) node {probability \( p_v \)};
\end{tikzpicture}
\begin{tikzpicture}
\input{gen_counts}
\path [font=\scriptsize,near start]
	(00) edge node [above] {2} (10)
	(00) edge node [below] {1} (11)
	(10) edge node [above] {1} (20)
	(10) edge node [above] {1} (21)
	(11) edge node [above] {1} (21)
	(11) edge node [below] {0} (22)
	(20) edge node [above] {0} (30)
	(20) edge node [above] {1} (31)
	(21) edge node [very near start,above] {1} (30)
	(21) edge node [above] {1} (31)
	(22) edge node [below] {0} (30)
	(22) edge node [below] {0} (31);
\path (22) +(0,-0.75) node {count \( \NumVisits_v, \NumVisits_v^r \)};
\end{tikzpicture}
\begin{tikzpicture}
\input{gen_errors}
\path [font=\scriptsize,near start]
	(00) edge node [above] {\( \frac12 \)} (10)
	(00) edge node [below] {\( \frac12 \)} (11)
	(10) edge node [above] {\( \frac14 \)} (20)
	(10) edge node [above] {\( \frac14 \)} (21)
	(11) edge node [pos=0.4,below] {\( \frac14 \)} (21)
	(11) edge node [below] {\( \frac34 \)} (22)
	(20) edge node [pos=0.5,above] {\( \frac38 \)} (30)
	(20) edge node [pos=0.4,above] {\( \frac58 \)} (31)
	(21) edge node {\( \frac14 \)} (30)
	(21) edge node {\( \frac14 \)} (31)
	(22) edge node [pos=0.5,below] {\( \frac38 \)} (30)
	(22) edge node [pos=0.5,below] {\( \frac38 \)} (31);
\path (22) +(0,-0.75) node {error \( e_v, e_v^r \)};
\end{tikzpicture}
\caption{\label{fig:probsAndCounts}%
	Top row: simulated random walks, using one bit of catalytic space per vertex (shown as \( \uparrow \) or \( \downarrow \)) to alternate between random choices.
	Bottom row: a comparison to a true random walk, illustrating the parts of \cref{def:probsAndCounts}:
	visit probabilities \( p_v \) (left), visit \( \NumVisits_v \) and transition \( \NumVisits_v^r \) counts (centre), and errors \( e_v = |\NumVisits_v - 3 p_v| \), \( e_v^r = |\NumVisits_v^r - \frac32 p_v| \) (right).}
\end{figure}

\subsection{Summary of Contributions}
From a complexity perspective, neither result places new problems in catalytic logspace. However, we view these algorithms as having two main advantages over prior work. First, the techniques are simple and clearly demonstrate the power of catalytic computation. Second, this simplicity allows us to give concrete bound on their resource usage (both catalytic space and time). We view \textit{algorithms} in catalytic space as worthy of further study; we have no reason to suspect our runtimes are optimal.

\section{Preliminaries}
We denote \( \N = \{ 0, 1, 2, \dotsc \} \) the natural numbers, and for \( n \in \N \) we denote \( [n] = \{ 0, 1, \dotsc, n-1 \} \) the natural numbers less than \( n \).

For a node \( v \) in a directed graph, \( \DIn(v) \) is the number of incoming and \( \DOut(v) \) is the number of outgoing edges.

We require a very weak bound on the number of paths in a graph:
\begin{fact}\label{fct:pathbound}
    For an $n$-vertex graph (possibly with self-loops), the number of length-$T$ paths between any two vertices is at most $P_n=n^T$.
\end{fact}

We assume basic familiarity with word-RAM and catalytic machines. For concreteness, we give a (not entirely formal) definition of catalytic RAM machines. 
\begin{definition}[Catalytic RAM Machine]\label{def:catrammach}
    We say $\mathcal{A}$ is a catalytic RAM machine that computes a function $f:\zo^*\ra\zo^*$ using $T(n)$ time, $S(n)$ workspace, and $W(n)$ catalytic space if it works as follows. The machine is given read-only access to $x$, read-write access to $S(|x|)$ bits of workspace, read-write access to a catalytic tape $R$ of length $W(|x|)$ in initial configuration $\tau$, and write-only access to an output tape.
    
    Furthermore, we allow the algorithm query access to the catalytic tape, in that it can read and write a specified bit in constant time after writing the index of this bit to a dedicated query tape. 
    For every $x$ and $\tau$, we have that the machine halts in at most $T(|x|)$ steps with $f(x)$ on the output tape and the catalytic tape restored to $\tau$.
\end{definition}
    Our definition is not powerful enough to allow analysis of runtimes without polylog factors from query overheads, and we do not attempt this.

\subsection{Catalytic Registers}\label{sec:catalyticRegisters}

It is often convenient for an algorithm to view its catalytic space as consisting of registers \( R_1, R_2, \dotsc \) for doing arithmetic over some modulus \( q \).
If \( q \) is a power of two, this is straightforward: allocate \( \log q \) bits per register.
However, our faster randomized connectivity algorithms (\cref{thm:mainNLrand,thm:mainNLrandrevert}) need to work with arbitrary moduli \( q \), which creates a difficulty: if we allocate \( \ell = \lceil \log q \rceil \) bits of the catalytic tape to each register, some registers may start with values outside the range \( 0, \dotsc, 2^{ \ell } - 1 \).

Buhrman, Cleve, Kouck\'y, Loff and Speelman~\cite{Catalytic} have a clever solution to this problem which unfortunately is too slow for our purposes.
Instead, we do the following.

Choose \( \ell \) to be a constant factor larger than \( \lceil \log q \rceil \), and treat the \( \ell \)-bit string as an element of $\Z/qd\Z$, where $d=d(q,\ell)$ is the largest value such that $qd\le 2^{ \ell }$.
We say $R$ is \emphdef{valid for (modulus) $q$} if its initial value is less than $qd$, and otherwise we say $R$ is \emphdef{invalid}.

Our algorithm ensures its registers are valid by applying a \emph{shift} to all registers simultaneously.
\begin{fact}\label{fct:shift}
    For an $ \ell $-bit register $R$ with initial configuration $\tau$, over a uniformly random shift $\beta \in [2^{ \ell }]$, $ ( \tau+\beta ) \bmod 2^{ \ell }$ is valid for $q$ with probability greater than $1-1/(d+1)$.
\end{fact}

Each valid register can be decomposed as having value \( aq+b \); to do arithmetic modulo \( q \), we leave \( a \) fixed and only change the \( b \) part.
It is easy to see these operations can be computed in simultaneous time $\tO( \ell )$ and space $O( \log \ell )$, given $q,d$ as input. For valid registers $R,R'$ over modulus $q$, we let $R\la R\pm R'$ to be this operation.

Now that we can implement registers over arbitrary moduli, the following straightforward lemma lets us use them to check whether counts over a much larger domain are nonzero:
\begin{lemma}\label{lem:divis}
    Let $V\le 2^P$ be arbitrary. For a uniformly random $r\in  [P^2]$, the probability that $V \equiv 0 \pmod r$ is at most $1-O(1/\log P)$.
\end{lemma}
\begin{proof}
    Note that $V$ can have at most $P$ distinct prime factors, and for every prime $q$ that is not one of these factors we have that $V \not\equiv 0 \pmod q$. Moreover, a random element in $[S]$ is prime with probability $\Omega(1/ \log S)$ by the Prime Number Theorem. Thus, if we consider the interval $[P^2]$, this interval contains at least $P^2/O(\log P)$ primes, of which at most $P$ divide $V$, so the probability that we draw a prime that does not divide $V$ is at least
    \[
    \frac{P^2/O(\log P)-P}{P^2} = \Omega(1/\log P) \ldotp \qedhere
    \]
\end{proof}

\subsection{Input Representation}\label{sec:InputRepr}
\newcommand{\vd}{\vec{d}}
\newcommand{\vs}{\vec{s}}

Because our catalytic algorithms do not have the space to perform otherwise standard transformations on the graph (for instance, producing an adjacency list given an adjacency matrix), we must be careful with how they access the input. We adopt the model of oracle access, which is common in sublinear and local models.
We say we have oracle access to a graph \( G \) if its vertex set is \( [n] \) for some \( n \in \N \), and for any \( v \in [n] \) we can make the following queries:
\begin{itemize}
    \item $\indeg_G(v)$ (resp.\ $\outdeg_G(v)$) returns the in-degree \( \DIn(v) \) (resp.\ out-degree \( \DOut(v) \)).
    \item $\innbr_G(v,i)$ (resp.\ $\outnbr_G(v,i)$) returns the $i$th in-neighbor (resp.\ out-neighbor) of $v$, or $\perp$ if this does not exist.
\end{itemize}
Our connectivity algorithms only use in-edge access to the graph, and our random walk algorithms only use out-edge access.

For the locally revertible connectivity algorithm (\cref{thm:mainNLrandrevert}), we use that given a graph, we can provide oracle access to a modified graph with bounded in-degree:
\begin{lemma}\label{lem:reducedegree}
    Given oracle access to a directed graph $G$ with $n$ vertices and $m$ edges, where every vertex has in-degree at most \( n \), there is a simulation of an oracle for a graph $G'$ on $n'=O(n^2)$ vertices with the following properties:
    \begin{itemize}
        \item The simulation can answer queries in $O(\log n)$ space and $O(\log n)$ time, with the exception of \( \outnbr_{ G' } \) queries.
        \item The maximum in-degree is $2$.
        \item The diameter is $\tO(n)$.
        \item For vertices $s,t \in [n]$, there is an $s\ra t$ path in $G$ if and only if there is an \( s \ra t \) path in \( G' \). (We have \( [n] \subseteq [n'] \), so integers \(s,t\) representing vertices of \( G \) also represent vertices of \( G' \).)
    \end{itemize}
    All but \( O( m ) \) of the nodes of \( G' \) are isolated (no in- or out-edges), and there is an algorithm to list all non-isolated nodes in \( O(\log n) \) space and \( \tO(m+n) \) time.
\end{lemma}
\begin{proof}
    We wish to replace each vertex of \( G \) with a binary tree, with edges pointing toward the root.
    An edge from \( u \) to \( v \) will be replaced with an edge from the root of \( u \) to one of the leaves of \( v \).

    Suppose a node \( v \in [n] \) has in-degree at least two, and its in-edges are from nodes \( u_1, \allowbreak \dotsc, \allowbreak u_{ \indeg_G(v) } \).
    Then \( v \) is represented by vertices in \( G' \) which we will call \( (v,0), \allowbreak (v,1), \dotsc, \allowbreak (v, \indeg_G(v)-2) \).
    We create edges to turn these nodes into a binary tree, and add \( ( u_1, 0 ), \allowbreak \dotsc, \allowbreak ( u_{\indeg_G(v)}, 0 ) \) as leaves, using heap indexing as follows.
    Every node \( (v,i) \) has in-degree two.
    The first in-edge of \( (v,i) \) comes either from \( (v,2i+1) \), or from the ``leaf node'' \( (u_{ 2i+3-\indeg_G(v) }, 0) \) if \( 2i+1 \ge \indeg_G(v)-1 \).
    Similarly, the second in-edge comes either from \( (v,2i+2) \) or from \( (u_{ 2i+4-\indeg_G(v) }, 0) \).
    See \cref{fig:DegreeTwoTransform}.

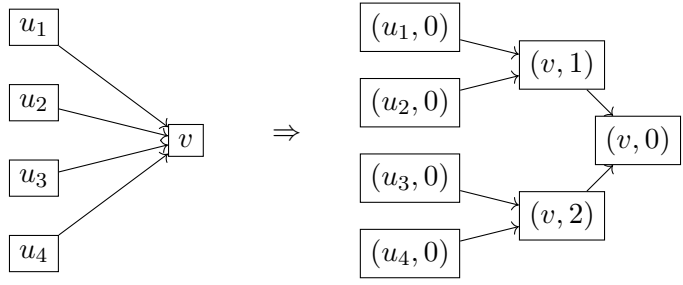
\begin{figure}
\centering
\begin{tikzpicture}[v/.style={draw},baseline=(v)]
\path
    +(0,0) node [v] (u1) {\( u_1 \)}
    +(0,-1) node [v] (u2) {\( u_2 \)}
    +(0,-2) node [v] (u3) {\( u_3 \)}
    +(0,-3) node [v] (u4) {\( u_4 \)}
    +(2,-1.5) node [v] (v) {\( v \)}
    (u1) edge [->] (v)
    (u2) edge [->] (v)
    (u3) edge [->] (v)
    (u4) edge [->] (v);
\end{tikzpicture}
\qquad\( \Rightarrow \)\qquad
\begin{tikzpicture}[v/.style={draw},baseline=(v0)]
\path
    +(0,0) node [v] (u1) {\( ( u_1, 0 ) \)}
    +(0,-1) node [v] (u2) {\( ( u_2, 0 ) \)}
    +(0,-2) node [v] (u3) {\( ( u_3, 0 ) \)}
    +(0,-3) node [v] (u4) {\( ( u_4, 0 ) \)}
    +(3,-1.5) node [v] (v0) {\( ( v, 0 ) \)}
    +(2,-0.5) node [v] (v1) {\( ( v, 1 ) \)}
    +(2,-2.5) node [v] (v2) {\( ( v, 2 ) \)}
    (u1) edge [->] (v1)
    (u2) edge [->] (v1)
    (u3) edge [->] (v2)
    (u4) edge [->] (v2)
    (v1) edge [->] (v0)
    (v2) edge [->] (v0);
\end{tikzpicture}
\caption{\label{fig:DegreeTwoTransform}Converting a graph to an equivalent one with in-degree bounded by \( 2 \).}
\end{figure}

    If \( v \) has in-degree less than two, it is represented by a single node \( (v,0) \).

    Our notion of oracle access requires vertices of \( G' \) to be represented as integers in some range \( [n'] \).
    We set \( n' = n(n-1) \).
    Every integer \( v' \in [n'] \) can be written as \( v' = v + ni \) for a unique vertex \( v \in [n] \) and index \( i \in [n-1] \).
    Define \( \varphi(v,n) := v+ni \).
    We identify \( v' \) with the vertex \( (v,i) \) described above, or, if \( i \) is too big (greater than \( \min \{ \indeg_G(v)-2, 0 \} \)), then \( v' \) is an isolated vertex with no in- or out-edges.
    (This means \( G' \) has more vertices than strictly necessary, but it makes it easy to interpret indices in \( [n'] \) as nodes.)

    We now verify we can provide oracle access to this new graph.
    All of the oracle routines \( \indeg_{ G' }, \outdeg_{ G' }, \innbr_{ G' }, \outnbr_{ G' } \) first decompose their input as \( v' = \varphi(v,i) \).
    They return \( 0 \) or \( \perp \) if \( i > \min \{ \indeg_G(v)-2,0 \} \).
    Otherwise:
    \begin{itemize}
        \item \( \indeg_{ G' }( \varphi(v,i) ) \) is is \( 2 \), unless \( \indeg_G(v) < 2 \) in which case \( \indeg_{ G' }( (v,1) ) = \indeg_G(v) \).
        \item $\outdeg_{G'}( (v,0) ) = \outdeg_G(v)$.
        \item For \( i>0 \), $\outdeg_{G'}( \varphi(v,i) ) = 1$.
        \item To compute \( \innbr_{ G' }( \varphi(v,i), j ) \), first compute \( d = \indeg_G(v) \), and return \( \perp \) if \( j > d \). Otherwise, let \( k = 2i+1+j \). If \( k \ge d \), return \( \varphi(u,0) \) where \( u = \innbr_G( v, k+2-\indeg_G(v) ) \); otherwise, return \( \varphi(v,k) \).
         
        \item One could implement $\outnbr_{G'}( \varphi(v,i), j )$ by using \( \indeg_{G'} \) and \( \innbr_{G'} \) to exhaustively find all edges out of \( (v,i) \), and then return the \( j \)-th in lexicographic order. (We do not actually use this operation on the graph \( G' \).)
        
    \end{itemize}
    Now, we verify $G'$ has the desired properties.
    The maximum in-degree is immediate from the construction.
    For every edge \( u,v \) in \( G \), there is a length-\( O( \log \indeg_G( v ) ) \) path from \( \varphi(u,0) \) to \( \varphi(v,0) \) in \( G' \), since in the tree representing node \( v \), every node has a logarithmic-length path to the root \( (v,0) \).
    It follows that if there is a path from \( s \) to \( t \) in \( G \), there is a path from \( \varphi( s, 0 ) = s \) to \( \varphi( t, 0 ) = t \) in \( G' \).
    Conversely, if we project every vertex \( (v,i) \) to the corresponding vertex \( v \) in \( G \), then every edge in \( G' \) is projected either to a self-loop or an edge that exists in \( G \), and so an \( s \to t \) path in \( G' \) implies an \( s \to t \) path in \( G \).
    To see that the diameter is \( O( n \log n ) \), observe that every path \( p \) in \( G' \) consists of two kinds of two kinds of edges: edges \( (v,i) \to (v,i') \) within the same tree, and edges \( (v,0) \to (v',i) \) between trees. Since the trees have depth \( O( \log n ) \), there can never be more than that many edges within trees in a row, and so \( p \) corresponds to a path a factor of at most \( O( \log n ) \) shorter in \( G \).
    Since the diameter of \( G \) is \( O( n ) \), the diameter of \( G' \) is \( O( n \log n ) \).

    Finally, all but \( O(m) \) nodes in \( G' \) are isolated, since \( G' \) has \( O(m) \) edges.
    These can be enumerated efficiently by considering the nodes \( v \in [n] \) one at a time.
\end{proof}

\section{Deciding Graph Connectivity}\label{sec:resultNL}

We state our first algorithm, which requires one register per vertex and timestep (but permits local revertibility).

\newcommand{\len}{T}
\begin{theorem}\label{thm:stcount}
    There is a catalytic logspace algorithm that, given $\ell,q,\len\in \N$ and a graph $G$ with $n$ vertices and $m$ edges and vertices $s,t$, and $\ell$-bit registers $R_{i,u}$ for $i\in \{0,\ldots,\len\}$ and $u\in V$ that are valid for modulus $q$ (\cref{sec:catalyticRegisters}), returns
    \[
    ( \# \text{of $s\ra t$ paths of length $\len$} ) \bmod q \ldotp
    \]
    Moreover, the algorithm runs in time $\tO(\ell^2\cdot  \len \cdot (m+n))$, and is locally revertible in time $\tO(\ell\cdot d_{\max})$, where $d_{\max}$ is the maximum in-degree of $G$.
\end{theorem}
\begin{proof}
    We first implicitly lift $G = (V,E)$ to a layered graph on $(n+1)n$ vertices, defined as follows. We let the vertex set be $\{0,\ldots,n\}\times V$, and the edge set be
    \[E' = \{((i,v),(j,u)):j=i+1 \text{ and } (v,u)\in E\}.
    \]
    For every vertex $a=(i,v)$, let $R_a$ be the corresponding register with initial value $\tau_a$.

    We then describe the basic algorithm. For an edge $a=(i,u),b=(i+1,v)$, we let an \emphdef{edge push (resp.\ reverse edge push)} be the update where we set 
    \[R_b\la R_b+R_a\quad  \text{(resp.\ }R_b\la R_b-R_a\text{)}\] 
    where the arithmetic operations are defined as in \cref{sec:catalyticRegisters}.
    
    We let $\lp_i$ (resp.\ $\rp_i$) be the operation where we perform an edge push (resp.\ reverse edge push) on every edge from layer $i$ to layer $i+1$. We do this by iterating over vertices $v$, and using $\innbr_G(v)$ oracle queries to determine the in-neighbors of $v$, and then pushing along these edges. 
    
    Finally, let $\is(b)$ be the operation where we set
    \[
    R_{(0,s)} \la R_{(0,s)}+b.
    \]
    For each $b$, we define the push and reverse sequence
    \[
    \cP_b=(\is(b),\lp_0,\ldots,\lp_{\len-1})
    \]
    \[
    \cR_b = (\rp_{\len-1},\ldots,\rp_0,\is(-b)).
    \]
    First, note that the catalytic tape is easy to reset:
    \begin{claim}\label{clm:reset}
        For every value $b$ and register $R_a$, after executing $(\cP_b,\cR_b)$ we have that $R_a=\tau_a$.
    \end{claim}
    \begin{proof}
        It clearly suffices to prove that $\lp_{i},\rp_{i}$ preserve the tape configuration. This follows directly from the linearity of addition and the definition of both operations.
    \end{proof}
    Finally, the difference in the final values at a register with $b=\{0,1\}$ is exactly the number of paths from $s$ to this register:
    \begin{lemma}\label{lem:pushcount}
        For every $i,v,b$, let $\alpha_{(i,v),b}$ be the value of $R_{(i,v)}$ after $\cP_b$. Then $ \alpha_{(i,v),1}-\alpha_{(i,v),0}$ is exactly the number of length $i$ $s\ra v$ paths in $G$, modulo $q$.
    \end{lemma}
    \begin{proof}
        Note that $\alpha_{(i,v),b}$ is the value of $R_{(i,v)}$ after
        \[
        \is(b),\lp_0,\ldots,\lp_{i-1}
        \]
        since subsequent operations in $\cP_b$ do not write to $R_{(i,v)}$. 
    
        Suppose the claim holds for every vertex in layer $i-1$. Next, fix an arbitrary vertex $w=(i,v)$, and let $(u_1,\ldots,u_r)\subseteq [V]$ be the in-neighbors of $v$ in $G$. Note that every length-$i$ path from $s$ to $v$ decomposes as
        \[
        (s,\ldots,u_j)(u_j,v)
        \]
        for some unique $u_j$, so the paths are in one-to-one correspondence with length $i-1$ paths from $s$ to $u_j$ for some $j$.
    
        Finally, recall that $w$ has in-neighbors
        \[a_1=(i-1,u_1),\ldots,a_r=(i-1,u_r).\]
        and observe that
        \[
        \alpha_{w,b} \equiv \tau_w + \sum_{j\in [r]}\alpha_{a_j,b} \pmod{q}
        \]
        and hence
        \[
        \alpha_{w,1}-\alpha_{w,0} \equiv \sum_{j\in [r]}(\alpha_{a_j,1}-\alpha_{a_j,0}) \pmod{q}
        \]
        and so by the inductive hypothesis we are done.
    \end{proof}
    Then the final algorithm is straightforward. We determine and print $\alpha_{(n,t),1}-\alpha_{(n,t),0} \pmod{q}$. We compute this value by comparing the $\ell$-bit registers bit by bit, each time using the sequence $\cP_b,\cR_b$ with alternating values of $b$. 

    \begin{lemma}
        The algorithm runs in $\tO(\ell^2\len \cdot (n+m))$ time.
    \end{lemma}
    \begin{proof}
        Given an edge $e$ and layer $i$, pushing along this edge takes time $\tO(\ell)$, and hence a layer push takes time $\tO((m+n)\cdot \ell)$. Therefore executing $\cP_b$ and $\cR_b$ takes time $\tO(\len \cdot (n+m) \ell)$. Finally, we invoke both routines $\ell$ times to compute the final value, so the total runtime is as claimed. 
    \end{proof}

    Finally, we show how the algorithm is locally revertible.
    \begin{lemma}\label{lem:STConnLR}
        The algorithm is locally revertible in time $\tO(\ell\cdot d_{\max})$.
    \end{lemma}
    \begin{proof}
        Suppose we receive a query to return $\tau_a$, the initial value of the register $R_{a=(i,v)}$. Let $u_1,\ldots,u_r$ be the in-neighbors of $v$, which we can enumerate over using the oracle for $G$. The register $R_a$ is either in its initial state (in which case we can return its current value without modifying the tape), or some push sequence $\cP_b$ has been executed. In that case, the current value of $R_a$ is
        \begin{align*}
            R_a &= \alpha_{a,b}\\
            &= \tau_a + \sum_{j\in [r]}\alpha_{(i-1,u_j),b}\\
            &= \tau_a + \sum_{j\in [r]}R_{(i-1,u_j)}
        \end{align*}
        where the second equality follows from the definition of $\cP_b$, and the third follows from the fact that we do not modify registers in layer $i-1$ before reverting layer $i$ to the original register configuration. Thus, we can recover $\tau_a$ by enumerating over the in-neighbors of $v$ and computing
        \[
        R_a - \sum_{j\in [r]}R_{(i-1,u_j)} = \tau_a.
        \]
        Afterwards, we revert $R_a$ to $\alpha_{a,b}$ and continue execution as before. The time for the query is bounded by $\tO(\ell\cdot r)=\tO(\ell\cdot d_{\max})$ as claimed. 
    \end{proof}
    This completes the proof of the desired properties.
\end{proof}

Next, we present algorithm version that avoids lifting the graph, at the cost of not permitting local revertibility. In this case, we do not compute the number of $s\ra t$ paths mod the register size, simply a number that is nonzero if a path exists.
\begin{theorem}\label{thm:stnonzero}
    For every $\len\in \N$ and graph $G$ with $n$ vertices and $m$ edges and vertices $s,t$, there is $\zeta_{G,s,t}\in \N$ where:
    \begin{itemize}
        \item $\zeta_{G,s,t}\le (n+1)^\len$, and
        \item $\zeta_{G,s,t}$ is nonzero if and only if there is an $s\ra t$ path in $G$ of length at most $\len$.
    \end{itemize}
    Moreover, there is a catalytic logspace algorithm that, given $\ell,q,\len\in \N$ and a graph $G$ and vertices $s,t$, and $\ell$-bit registers $R_{\sigma,v}$ for $\sigma\in \zo$ and $v\in V$ that are valid for modulus $q$, returns 
    \[
    \zeta_{G,s,t} \bmod q \ldotp
    \]
    Moreover, the algorithm runs in time $\tO(\ell^2\cdot  \len \cdot (n+m))$.
\end{theorem}
\begin{proof}
    We describe the basic algorithm, which is like that of \cref{thm:stcount} but uses fewer registers.
    
    We define layer push operations given a parity value $\sigma\in \zo$. For $(u,v)\in E$ in some fixed order (where we add dummy edges $(u,u)$ for every $u$), set
    \[
    R_{\neg \sigma,v}\la R_{\sigma,v} + R_{\sigma,u}.
    \]
    and then set $\sigma \la \neg \sigma$.
    We define a reverse layer push as, for the same set of edges, setting
    \[
    R_{\neg\sigma,v}\la R_{\sigma,v} - R_{\sigma,u}
    \]
    and $\sigma\la \neg \sigma$.
    
    Next let $\is(b)$ be the operation where we set
    \[
    R_{0,s} \la R_{0,s}+b.
    \]
    For each $b$, we initialize $\sigma=0$ and define the push and reverse sequence:
    \[
    \cP_b=(\is(b),\lp^{(n)}), \qquad \cR_b = (\rp^{(n)},\is(-b)).
    \]
    
    By essentially the same argument as \cref{clm:reset}, we have that after executing $(\cP_b,\cR_b)$ for $b\in \zo$, every register $R_{\sigma,v}$ is reset to its initial configuration $\tau_{\sigma,v}$. The runtime is straightforward from the description.

    Finally, we must prove correctness. For every $i$, let $\sigma_i\in \zo$ be the parity of the registers pushed to in phase $i$ (and note that $\sigma_0=0$).
    \begin{definition}
        For every $(\sigma,v),i$, define $\zeta_{(\sigma,v),i}$ recursively as follows. First, set $\zeta_{(0,s),0}=1$ and for $v\neq s$ let $\zeta_{(0,v),0}=0$. Then for every $v$ define
        \[
        \zeta_{(\sigma_{i+1},v),i+1} := \zeta_{(\sigma_{i-1},v),i-1}+\sum_{u:(u,v)\in E}\zeta_{(\sigma_i,u),i}.
        \]
    \end{definition}
    We prove that these values are exactly the register difference after the pushes:
    \begin{lemma}
        Let $\alpha_{(\sigma_i,v),i,b}$ be the value of $R_{\sigma_i,v}$ after $\is(b),\lp^{(i)}$. For every $v\in V$ we have
        \[
        \alpha_{(\sigma_i,v),i,1}-\alpha_{(\sigma_i,v),i,0}\equiv \zeta_{(\sigma_i,v),i} \pmod q.
        \]
    \end{lemma}
    \begin{proof}
        For the base case, we have that
        \[
        \alpha_{(0,v),0,1}-\alpha_{(0,v),0,0} = \mathbb{I}[v=s] = \zeta_{(0,v),0}
        \]
        Now assume this holds for $i$ and $i-1$ and consider the $i+1$st push.
        WLOG suppose $\sigma_{i+1}=0$. For $v\in V$ we have
        \[
        \alpha_{(0,v),i+1,b} \equiv \alpha_{(0,v),i-1,b} + \sum_{u:(u,v)\in E}\alpha_{(1,u),i,b} \pmod q
        \]
        and hence 
        \[
        \alpha_{(0,v),i+1,1} - \alpha_{(0,v),i+1,0} \equiv \zeta_{(0,v),i-1} + \sum_{u:(u,v)\in E}\zeta_{(1,u),i} \equiv \zeta_{(1,v),i+1} \pmod q. \qedhere
        \]
    \end{proof}
    Then a simple inductive argument proves that $\zeta_{(\sigma_i,v),i}\le (n+1)^i$, and moreover that the set of $v$ for which $\zeta_{(\sigma_i,v),i}>0$ is exactly those $v$ with an $s\ra v$ path of length at most $i$.
\end{proof}

\subsection{Putting it all together}
We then use these results to prove the main theorems.

For the deterministic algorithm, we choose the register size and modulus $q$ large enough so that the count of paths mod $q$ is equal to the count of paths. 
\mainNL*
\begin{proof}
    We initialize $2n$ registers $\{R_{\sigma,v}\}_{\sigma\in \zo,v\in V}$ each of size $\ell=\lceil \log (n+1)^n+1\rceil$ and set $q=2^\ell$. Since we chose $q=2^\ell$, all registers are valid no matter their initial configuration, so we immediately invoke \cref{thm:stnonzero} with $\len=n$. If the value obtained is nonzero, we return that there is a path, and otherwise return that there is no path. The runtime is immediate from the choice of $\ell$ and $\len$.
\end{proof}

We now give the randomized algorithm.
{\renewcommand\footnote[1]{}
\mainNLrand*}
\begin{proof}
    Let $B_n=(n+1)^n+1$.
    \paragraph{The Algorithm}
    For $I=O(\log n)$ iterations (where the specific constant is to be chosen later), we proceed as follows. We initialize $2n$ registers $\{R_{\sigma,v}\}_{\sigma \in \zo, v\in V}$, each of size
    \[\ell=5\lceil \log B_n\rceil.\] 
    We draw a random modulus $q\in [\log^2 B_n]$ and store in on the worktape. Let $d$ be the largest value such that $qd\le 2^\ell$ (which we can compute in time $\polylog(n)$ and store on the worktape). We have that
    \[
    d \ge \frac{2^{ \ell }}{2q}>n^3
    \]
    Next, we draw a random shift $\beta \in [2^{ \ell }]$ and store it on the worktape. We add \( \beta \) to each register $R_{\sigma,v}$ with initial configuration $\tau_{\sigma,v}$ and verify that $\tau_{\sigma,v}+\beta$ is valid for $q$. If not, we subtract \( \beta \) from all registers and abort (and return $\perp$). Otherwise, we invoke \cref{thm:stnonzero} with this register set and $\len=n$.
    Let the returned value be $\zeta$. We first subtract $\beta$ from all registers. Then, if $\zeta\neq 0$, we return that there is a path, and otherwise proceed to the next iteration. If we exhaust all iterations, we return that there is not a path. 

    \paragraph{Success Probability}
    We first argue that the algorithm does not abort with high probability. There are $2n$ registers, each of which we shift $O(\log n)$ times. By \cref{fct:shift}, the probability that any such shift results in an invalid configuration is at most $2/n^3$, so a union bound completes the proof. 
    
    Next, note if $\zeta_{G,s,t}=0$, the algorithm clearly never returns there is a path. Otherwise, for each $q$ drawn by the algorithm, by \cref{lem:divis} the probability that $\zeta_{G,s,t} \equiv 0 \pmod p$ is at most $(1-O(1/\log\log B_n))=(1-1/O(\log n))$, and hence choosing $I=O(\log n)$ sufficiently large, we obtain that with probability at least $1-1/100$ there is some iteration where we detect a nonzero number of paths, and thus succeed.

    Finally, runtime and space consumption follow directly from the description of the algorithm.
\end{proof}

Finally, we finish the proof of the locally revertible algorithm.
{\renewcommand{\footnote}[1]{}%
\mainNLrandrevert*}
\begin{proof}
    We first modify the input graph $G$ by simulating access to the graph $G'$ using the query oracle of \cref{lem:reducedegree}. Let \( n' = O(n^2) \) be the number of vertices of \( G' \), let $n''=O(m)$ be the number of non-isolated vertices, let $\len=\tO(n)$ be its diameter, and recall that (after virtually adding a self-loop on $t$) it has maximum in-degree $3$. Let $p_{n'}=\lceil \log P_{n'}\rceil$ be the logarithm of the bound on the number of paths in $G'$ of \cref{fct:pathbound}.

    \paragraph{The Algorithm}
    For $I=O(\log n)$ iterations (where the specific constant is to be chosen later), we proceed as follows.
    First, we initialize the registers we will use by adding a random shift as described in \cref{sec:catalyticRegisters}.
    We will have a register corresponding to each layer and vertex in $G'$, indexed as
    \[ (i,v)\in \{0,\ldots,\len\}\times [n'] \ldotp \]
    We allocate $n^3$ total registers on the catalytic tape corresponding to these nodes, but we only initialize the \( \len n'' = \tO(nm) \) registers which corresond to non-isolated nodes of \( G' \).
    (Recall \cref{lem:reducedegree} gives an algorithm for enumerating these vertices in \( \tO(m+n) \) time, amounting to \( \tO(n(m+n)) \) time for \( n \) layers.)
    Call a register \emph{relevant} if it corresponds to \( (i,v) \) where \( v \) is a non-isolated vertex or \( v \in \{ s,t \} \).

    Each register has size
    \[ \ell=5\lceil \log p_{n'}\rceil \ldotp \] 
    We draw a random modulus $q\in [p_{n'}^2]$ and store in on the worktape. Let $d$ be the largest value such that $qd\le 2^\ell$ (which we can compute in time $\polylog(n)$ and store on the worktape). We have that
    \[
    d \ge \frac{2^{ \ell }}{2q}\ge \frac{{m}^3}{2}.
    \]
    Next, we draw a random shift $\beta \in [2^{ \ell }]$ and store it on the worktape. We add \( \beta \) to each relevant register $R_a$ with initial configuration $\tau_a$, and then verify that each $\tau_a+\beta$ is valid for $q$. If not, we first subtract $\beta$ from all relevant registers, and then abort (and return $\perp$). Otherwise, we invoke \cref{thm:stcount} with
    \[
    G=G',\qquad \len=\len
    \]
    except we modify it to only push values from relevant registers.
    Let the returned value be $\alpha$. We restore the catalytic tape by subtracting $\beta$ from all relevant registers. Then, if $\alpha\neq 0$, we return there is a path, and otherwise proceed to the next iteration. If we exhaust all iterations, we return that there is not a path. 

    \paragraph{Success Probability}
    We first argue that the algorithm does not abort with high probability. Note that there are $\tO(n m)$ registers, each of which we shift $O(\log n)$ times. By \cref{fct:shift}, the probability that any such shift results in an invalid configuration is at most $4/m^3$, so a union bound completes the proof. 
    
    Next, let $V$ be the total number of length-$n$ paths from $s$ to $t$ in the modified graph. If $V=0$, the algorithm clearly never returns there is a path. Otherwise, for each $q$ drawn by the algorithm, by \cref{lem:divis} the probability that $V \equiv 0 \pmod p$ is at most $(1-O(1/\log p_{n'}))=(1-1/O(\log n))$, and hence choosing $I=O(\log n)$ sufficiently large, we obtain that with probability at least $1-1/100$ there is some iteration where we detect a nonzero number of paths, and thus succeed.

    Finally, runtime and space consumption follow directly from the description of the algorithm.

    \paragraph{Local Revertibility}
    Finally, we argue that the algorithm is locally revertible. Given a query on register $a$, we query the local revertibility routine of \cref{thm:stcount} on this register. Since we initialized the algorithm of \cref{thm:stcount} with $R_a$ in configuration $\tau_a+\beta$, it returns this value, and we return $\tau_a$ by subtracting the stored shift $\beta$. Finally, since $G'$ has constant degree, the time is as claimed.
\end{proof}

\section{Estimating Random Walks}\label{sec:resultBPL}

The bulk of our proof of \cref{thm:randomGeneral} is a technique for simulating random walks on \emph{acyclic} graphs:
\begin{theorem}[random walk on a DAG]\label{thm:randomDag}
There is an algorithm which, given a directed acyclic graph \( G \) with \( n \) vertices and \( m \) edges, together with vertex \( s \), sink vertex \( t \), and parameter \( \eps > 0 \), returns \( \rho \) such that
\[
\left|\rho - \Pr[\text{random walk from $s$ reaches $t$}]\right|\le \eps \ldotp
\]
The algorithm runs in time \( \tO( nm / \eps ) \).
It uses \( O(\log(nm/\eps)) \) workspace and \( \tO(n \log( m / \allowbreak \eps ) ) \) catalytic space; the algorithm is guaranteed to restore the catalytic tape to its initial state as long as the input is valid --- in particular, the graph \( G \) has no cycles.
\end{theorem}
(This version of our algorithm has the curious property that it can be tricked into making irreversible changes to its catalytic tape if \( G \) has cycles. If this is undesirable, the algorithm could be changed to require an efficiently checkable proof that \( G \) is acyclic, like a topological ordering of its vertices.)

The algorithm is described by \( \AlgWalk \) (\cref{alg:walk}) and its subroutine \( \AlgWalkOnce \) (\cref{alg:walkOnce}).
The prove \cref{thm:randomDag}, we must prove two things for every $G,s,t,\eps$: in \cref{sec:correct}, we show that \( \AlgWalk( G,s,t,\eps ) \) gives the correct output, and in \cref{sec:restore}, we show that it restores the catalytic tape at the end of the computation.
\Cref{sec:completeRandomDagProof} ties the proof together and analyzes the time and space used.
\Cref{sec:completeRandomGeneralProof} proves \cref{thm:randomGeneral} by converting any graph into a larger acyclic graph, and then in \cref{sec:stationary}, we explore what happens if we apply our technique directly to a graph with cycles, skipping the conversion step.

\subsection{Registers}\label{sec:walkRegisters}

The algorithm uses one register \( R_v \) on the catalytic tape for every vertex \( v \) of the graph.
We allocate \( \ell \) bits to each register, where \( \ell = \lceil \log K \rceil \) and
\( K = \KVALUE \)
is the number of simulations run by \cref{alg:walk}.
So, each register stores a number in the range \( 0, \dotsc, 2^{ \ell } - 1 \).
We choose this value of \( \ell \) so that if we increment the value of a register \( K \) times, each time adding one modulo \( 2^{ \ell } \), there is at most one time when it resets to \( 0 \) instead of increasing by one.
Concretely, this is used in the proof of \cref{lem:fair} to bound the error introduced by this reset.

\begin{algorithm}
Let \( K = \KVALUE \).\;
\( \NReach \gets 0 \)\;
\tcc{Forward phase: estimate probability of reaching \( t \)}
\For{\nllabel{line:beginForwardLoop}\( i \in [K] \)}{
\( v \gets \AlgWalkOnce( M, x, s, \Fwd ) \)\;
\If{\( v = t \)}{
\( \NReach \gets \NReach + 1 \)\;
}
}\nllabel{line:endForwardLoop}
\tcc{Reverse phase: reset the catalytic tape}
\For{\( i \in [K] \)}{
\( \AlgWalkOnce( M, x, s, \Rev ) \)\;
}
\Return{$\NReach/K$}\label{line:ReturnNReach}
\caption{\label{alg:walk}\( \AlgWalk(G,s,t,\eps) \). Parameters: graph $G$, source $s$, target $t$, accuracy $\eps$.)}
\end{algorithm}
\begin{algorithm}
Registers: one register \( R_v \) in \( [2^{ \ell }] \) for every vertex \( v \); see \cref{sec:walkRegisters}.\;
\( v \gets s \)\;
\While{\nllabel{line:checkHalt}\( \outdeg(v) > 0 \)}{
	\tcc{Choose an outgoing edge based on \( R_v \), and update \( R_v \).}
	\uIf{\( m = \Fwd \)}{
		\( r \gets R_v \bmod{\outdeg(v)} \)\label{line:chooseForward}\;
		\( R_v \gets (R_v + 1) \bmod 2^{ \ell } \)\nllabel{line:addForward}\;
	}
	\Else(\tcc*[h]{\( m = \Rev \)}){
		\( R_v \gets (R_v - 1) \bmod 2^{ \ell } \)\nllabel{line:subtractReverse}\;
		\( r \gets R_v \bmod{\outdeg(v)}\)\;
	}
	\( v \gets \outnbr_G(v,r) \)\nllabel{line:walkForward}\;
}
\Return{\( v \)}\;
\caption{\label{alg:walkOnce}\( \AlgWalkOnce(G,s,m) \). Parameters: graph $G$, source $s$, mode \( m \in \{ \Fwd, \Rev \} \). Returns a sink vertex of \( G \). (This algorithm modifies the catalytic tape, so it is not a catalytic algorithm. However, the changes are reversible, so it can be used as a subroutine in a catalytic algorithm.)}
\end{algorithm}

\subsection{The output is correct}\label{sec:correct}

To evaluate how well \( \AlgWalk \) simulates a random walk on \( G \), we will compare the number of times \( \AlgWalkOnce \) visits each vertex to the probability a true random walk would reach it.
We will argue that for every vertex \( v \), the algorithm splits its time fairly over the \( \DOut(v) \) different outgoing edges when it takes a step on \cref{line:walkForward}, from which it will follow that our simulation is sufficiently accurate.

Throughout this section, fix $G,s,t$ and $\eps$.
Fix an initial content of the catalytic tape $\tau$, and let $R_v$ be the \( \ell \)-bit register allocated to vertex $v$ with initial value $\tau_v$.

Let \( K = \KVALUE \) as in \( \AlgWalk \). 

The following definition establishes some notation to help reason about the accuracy of the simulation.
It is illustrated in \cref{fig:probsAndCounts}.
\begin{definition}[visit probability \( p_v \), visit count \( \NumVisits_v \), error \( e_v \), transition count \( \NumVisits_v^r \), transition error \( e_v^r \)]\label{def:probsAndCounts}
Let \( p_v \) be the probability that a random walk starting at \( s \) reaches a vertex \( v \).

Let \( \NumVisits_v \) be the number of times \( \AlgWalkOnce \) visits \( v \) during the forward phase of \( \AlgWalk \) (lines \ref{line:beginForwardLoop}--\ref{line:endForwardLoop}).
(A ``visit'' to the vertex stored in the variable \( v \) occurs whenever \( \AlgWalkOnce \) evaluates the while loop condition on line~\ref{line:checkHalt}.)

Define the \emph{error} \( e_v = |\NumVisits_v - K p_v| \).

For \( r \in [\DOut(v)] \), let \( \NumVisits_v^r \) be the number of those visits for which the variable \( r \) had the given value on \cref{line:walkForward} of \( \AlgWalkOnce \) (so \( \NumVisits_v = \sum_{ r=0 }^{ \DOut(v)-1 } \NumVisits_v^r \) unless \( v \) is a sink).
\( \NumVisits_v^r \) counts transitions where the algorithm followed the \( r \)-th outgoing edge from \( v \).

Define the \emph{transition error} \( e_v^r = | \NumVisits_v^r - K p_v / \DOut(v) | \).
\end{definition}

\begin{lemma}\label{lem:fair}
For every vertex \( v \) and \( r \in [\DOut(v)] \), \( e_v^r \le 2 + e_v / \DOut(v) \).
\end{lemma}
\begin{proof}
Let us temporarily imagine the register \( R_v \) always stores a value in \( [ \DOut(v) ] \), and that each time \( \AlgWalkOnce \) visits \( v \), \cref{line:addForward} increments it as
\( R_v \gets (R_v + 1) \bmod \DOut(v) \)
instead of incrementing modulo \( 2^{ \ell } \).
% Now, every time \( \AlgWalkOnce \) visits \( v \), it updates \( R_v \).
This causes the value \( r \) to cycle through the values \( 0, \dotsc, \DOut(v) - 1 \), so that on the \( t \)-th visit to \( v \), \( r = (\tau_v + t - 1) \bmod \DOut(v) \), where \( \tau_v \) is the starting value of \( R_v \) from the catalytic tape.
As a result, any two transition counts must differ by at most one:
\begin{equation}\label{eq:FantasyNumVisits}
| \NumVisits_v^r - \NumVisits_v^{ r' } | \le 1 \rlap{\qquad \text{(pretending \( R_v \in [\DOut(v)] \)).}}
\end{equation}

In fact, \( R_v \) cycles through the values \( 0, \dotsc, 2^{ \ell } \).
As a result \( r \) cycles through the values \( 0, \dotsc, \DOut(v) - 1 \) in order, with one exception: if register \( R_v \) cycles from \( 2^{ \ell - 1 } \) to \( 0 \), then the cycle is interrupted.
Since \( 2^{ \ell } \ge K \), this can happen at most once.
% NB: this only-resetting-once is referred to from sec:walkRegisters. So, update that if this argument gets moved.
So, instead of \cref{eq:FantasyNumVisits}, we have have that
\[ | \NumVisits_v^r - \NumVisits_v^{ r' } | \le 2 \ldotp \]

Since \( \NumVisits_v = \sum_{ r=0 }^{ \DOut(v)-1 } \NumVisits_v^r \), it follows that
\[
    \left| \NumVisits_v^r - \frac{ \NumVisits_v }{ \DOut(v) } \right|
    =
    \left| \NumVisits_v^r - \sum_{ r'=0 }^{ \DOut(v)-1 } \frac{ \NumVisits_v^{ r' } }{ \DOut(v) } \right|
    \le
    \frac{ 1 }{ \DOut(v) }\sum_{ r'=0 }^{ \DOut(v)-1 } | \NumVisits_v^r - \NumVisits_v^{ r' } |
    <
    2
    \ldotp
\]
and so
\[
e_v^r
= \left| \NumVisits_v^r - \frac{ K p_v }{ \DOut(v) } \right|
\le \left| \NumVisits_v^r - \frac{ \NumVisits_v }{ \DOut(v) } \right| + \left| \frac{ \NumVisits_v }{ \DOut(v) } - \frac{ K p_v }{ \DOut(v) } \right|
< 2 + \frac{ e_v }{ \DOut(v) }
\ldotp \qedhere
\]
\end{proof}

\begin{lemma}\label{lem:addError}
Let \( (u_1, v), \dotsc, (u_d, v) \) be all edges incoming to a vertex \( v \), where each \( (u_i, v) \) is the \( r_i \)-th outgoing edge from \( u_i \).
Then \( e_v \le \sum_{ i=1 }^d e_{ u_i }^{ r_i } \).
\end{lemma}
\begin{proof}
Using the facts that \( \NumVisits_v = \sum_{ i = 1 }^d \NumVisits_{ u_i }^{ r_i } \) and \( p_v = \sum_{ i = 1 }^d p_{ u_i } / \DOut( u_i ) \), we have:
\[
e_v
= | \NumVisits_v - K p_v |
= \left| \sum_{ i = 1 }^d ( \NumVisits_{ u_i }^{ r_i } - K p_{ u_i } / \DOut( u_i ) ) \right|
\le \sum_{ i = 1 }^m | \NumVisits_{ u_i }^{ r_i } - K p_{ u_i } / \DOut( u_i ) |
= \sum_{ i = 1 }^d e_{ u_i }^{ r_i } \qedhere
\]
\end{proof}

\begin{lemma}\label{lem:accurate}
The final value \( \NReach / K \) returned by \( \AlgWalk \) on \cref{line:ReturnNReach} is within additive error \( \eps \) of the true visit probability \( p_t \).
\end{lemma}
\begin{proof}
We prove this from \cref{lem:fair,lem:addError} using an induction argument.

Let \( V \) be the set of vertices of \( G \).
Let \( v_0, v_1, \dotsc, v_n \) be a topological order on \( V \) ending with the sink vertex \( t \).
That is, for any edge \( (u,v) \), \( u \) appears before \( v \) in the order, and \( v_n = t \).

For \( i \in \{ 0, 1, \dotsc, n-1 \} \), consider the cut of \( G \) with vertices \( v_0, \dotsc, v_i \) on the left and \( v_{ i+1 }, \dotsc, v_n \) on the right.
We are interested in the total error over all the transitions which cross each such cut.

Let \( F_i \subseteq \{ (v,r) \mid v \in V, r \in [ \DOut(v) ] \} \) be the set of transitions which cross the cut.
That is, a pair \( ( v, r ) \) is in \( F_i \) if \( v \) is on the left of the cut and the \( r \)-th outgoing edge from \( v \) leads to a vertex on the right of the cut.
See \cref{fig:topoInduction}.

Let \( D_i = \sum_{ j=0 }^i \DOut(v) \): that is, \( D_i \) is the number of edges that originate from the left side of the cut (whether or not they cross the cut).

Let \( \sigma_i = \sum_{ ( v, r ) \in F_i } e_v^r \), recalling from \cref{def:probsAndCounts} that \( e_v^r = | \NumVisits_v^r - K p_v / \DOut(v) | \)).
We will show by induction that \( \sigma_i \le 2 D_i \) for all \( i \).

Base case: \( \sigma_0 = \sum_{ r=0 }^{ \DOut(v_0) } e_{ v_0 }^r \).
We know \( \NumVisits_{ v_0 } = K \) and \( p_{ v_0 } = 1 \), so \( e_{ v_0 } = 0 \) and so by \cref{lem:fair}, \( \sigma_0 \le 2 \DOut( v_0 ) \).

Induction step: Fix \( i \in [n-1] \) and assume \( \sigma_i \le 2(i+1) \).
Let \( ( v_{ j_1 }, v_i ), \dotsc, ( v_{ j_d }, v_i ) \) be the edges incoming to \( v_i \), where each \( ( v_{ j_k }, v_i ) \) is the \( r_k \)-th outgoing edge from \( v_{ j_k } \).
Those are the transitions that contribute to \( \sigma_i \) but not \( \sigma_{ i+1 } \), so we have \( \sigma_{ i+1 } = \sigma_i + \sum_{ r=0 }^{ \DOut( v_i ) - 1 } e_{ v_i }^r - \sum_{ k=1 }^d e_{ v_{ j_k } }^{ r_k } \).
From \cref{lem:fair}, we have
\[ \sum_{ r=0 }^{ \DOut( v_i ) - 1 } e_{ v_i }^r \le 2 \DOut( v_i ) + e_{ v_i }, \]
and so applying \cref{lem:addError} gives
\[
\sigma_{ i+1 }
\le \sigma_i + 2 \DOut( v_i ) + e_{ v_i } - \sum_{ k=1 }^d e_{ v_{ j_k } }^{ r_k }
\le \sigma_i + 2 \DOut( v_i )
\le 2 D_{ i+1 }
\]
completing the induction.

\begin{figure}
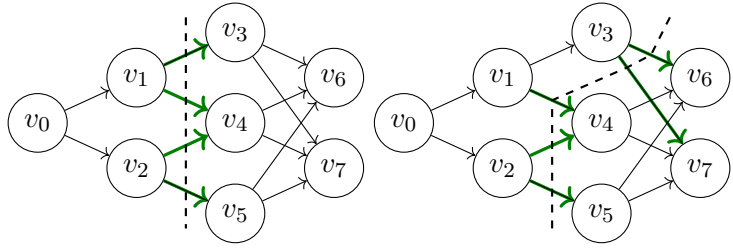

\centering
\begin{tikzpicture}
\input{gen_frontier_0}
\path (10) edge coordinate (a) (20);
\path (11) edge coordinate (b) (22);
\path (a) +(0,0.5) coordinate (top) (b) +(0,-0.5) coordinate (bottom);
\draw [dashed,thick] (top) -- (bottom);
\end{tikzpicture}
\begin{tikzpicture}
\input{gen_frontier_1}
\path
	(20) edge coordinate (a) (30)
	(20) edge coordinate [very near start] (b) (31)
	(10) edge coordinate (c) (21)
	(11) edge coordinate (d) (22)
	(a) +(0.25,0.5) coordinate (start)
	(d) +(0,-0.5) coordinate (end);
\draw [dashed,thick] (start) -- (a) -- (b) -- (c) -- (end);
\end{tikzpicture}
\caption{\label{fig:topoInduction}Two steps in the induction argument in the proof of \cref{lem:accurate}. On the left, \( i=2 \) and on the right, \( i=3 \). In each case, \( F_i \) consists of the edges that cross the dashed line, and \( \sigma_i \) is the sum of the errors on those edges.}
\end{figure}

In particular, \( \sigma_{ n-1 } \le 2 m \).
The corresponding set of edges \( F_{ n-1 } \) are exactly vertex \( t \)'s in-edges.
Therefore, by \cref{lem:addError}, \( e_t \) is at most
\[
\sum_{ ( v, r ) \in F_{ n-1 } } e_v^r
= \sigma_{ n-1 }
\le 2m
\]
When \( \AlgWalk \) reaches \cref{line:ReturnNReach}, \( \NReach = \NumVisits_t \), and so the algorithm returns \( \NumVisits_t / K \).
This is within \( \eps \) of \( p_t \) because
\[ \left| \frac{ \NReach }{ K } - p_t \right| = \frac{ e_t }{ K } \le \frac{ 2m }{ \KVALUE } \le \eps \ldotp \qedhere \]
\end{proof}

\subsection{The catalytic tape is restored.}\label{sec:restore}

\( \AlgWalk \) restores its catalytic tape when it finishes.
To see this, the following is enough:

\begin{lemma}\label{lem:reverse}
Running \( \AlgWalkOnce( G,s,t, \Fwd ) \) then \( \AlgWalkOnce( G,s,t, \Rev ) \) leaves all register values \( R_v \) with the same values they started with.
\end{lemma}
\begin{proof}
It is enough to show both calls to \( \AlgWalkOnce \) visit the same sequence of vertices, since then each register is incremented and decremented the same number of times.
(We say \( \AlgWalkOnce \) ``visits'' vertex $v$ each time the loop condition on line~\ref{line:checkHalt} is evaluated.)

Loosely speaking, this is true because each time \( \AlgWalkOnce(G, \allowbreak s, \allowbreak \Rev) \) decides which outgoing edge \( \outnbr_G(v,r) \) to follow, it first (\cref{line:subtractReverse}) undoes the change made by \( \AlgWalkOnce(G, \allowbreak s, \allowbreak \Fwd) \), and so it ends up choosing the same edge index \( r \).
This argument relies on the fact that \emph{a single run of \( \AlgWalkOnce \) never modifies the same register \( R_v \) more than once}, which follows from the fact that \( G \) has no cycles.

% Fix initial register values $\{\tau_v\}_{v\in V}$, and let $s=v_1,\ldots,v_k$ be the sequence of vertices visited by $\AlgWalkOnce(G,s,t,\Fwd)$, and let $s=u_1,\ldots,u_{k'}$ be the sequence visited by $\AlgWalkOnce(G,s,t,\Rev)$.

To make this more precise, let \( v_0, \dotsc, v_t \) be the vertices visited by \( \AlgWalkOnce( G, \allowbreak s, \allowbreak \Fwd ) \), and \( v'_0, \dotsc, v'_{ t' } \) be the vertices visited by the subsequent call to \( \AlgWalkOnce( G, \allowbreak s, \allowbreak \Rev ) \).
We show by induction that \( v_i = v'_i \) for each \( i \), so that in particular the main loop ends at the same sink vertex in each case and so \( t = t' \).

To begin with, \( v_0 = v'_0 = s \). Now assume \( v_i = v'_i \).
If \( v_i \) is a sink, both subroutine calls halt and we are finished.
Otherwise, we must show the same value \( r \) is chosen both times.
Since neither subroutine call made any other changes to \( R_v \), when the second subroutine call subtracts one from \( R_v \) on \cref{line:subtractReverse}, it exactly cancels out the only change made by the first subroutine call on \cref{line:addForward}, and so the same value \( r \) is recovered, and so the same next step is taken: \( v_{ i+1 } = v'_{ i+1 } \).
\end{proof}

\begin{corollary}\label{cor:isCatalytic}
\( \AlgWalk \) leaves its catalytic tape unchanged.
\end{corollary}

\begin{proof}
This is the same as saying the final values of all registers \( R_v \) equal the initial values.

By induction on \( K \), we can see that \( K \) calls to \( \AlgWalkOnce( G, s, \Fwd ) \) followed by \( K \) calls to \( \AlgWalkOnce( G, s, \Rev ) \) has no net effect on the registers \( R_v \).
For \( K = 0 \) this is clear.
\Cref{lem:reverse} provides the induction step.
That is, \( K+1 \) calls to each can be decomposed as (1) \( K \) calls to \( \AlgWalkOnce( G, s, \Fwd ) \), then (2) a single call to each, which by \cref{lem:reverse} has no net effect, then (3) \( K \) calls to \( \AlgWalkOnce( G, s, \Rev ) \).
Since (2) has no effect, we are left with \( K \) calls each, which by the induction hypothesis have no net effect.
\end{proof}

\subsection{Proof of \texorpdfstring{\cref{thm:randomDag}}{Theorem~\ref{thm:randomDag}} (random walk on a DAG)}\label{sec:completeRandomDagProof}

\Cref{lem:accurate} proves that \( \AlgWalk \) gives a correct answer, and \cref{cor:isCatalytic} proves that it restores its catalytic tape.
The runtime is dominated by \( 2K \) calls to \( \AlgWalkOnce \), each of which visits each vertex of \( G \) at most once.
Visiting a vertex means executing one iteration of the main loop of \( \AlgWalkOnce \), which takes \( \polylog (n+m) \) time, so the total run time is \( \tO( n K \polylog m ) = \tO( nm / \eps ) \).\footnote{%
It also uses \( \tO( n \polylog m ) \) time to count the number of edges \( m \) in order to compute \( K \).}
Each register takes \( O( \ell ) = O( \log ( m / \eps ) ) \) bits of the catalytic tape (\cref{sec:walkRegisters}), for a total of \( \tO( n \log( m / \eps ) ) \) catalytic space
The working memory only includes a constant number of variables \( \NReach \), \( v \), etc, each taking \( O( \log( nm / \eps ) ) \) working memory.
\qed

\subsection{Proof of \texorpdfstring{\cref{thm:randomGeneral}}{Theorem~\ref{thm:randomGeneral}} (random walk on a general graph)}\label{sec:completeRandomGeneralProof}

For convenience, we restate the theorem here:
\ThmRandomGeneral*
\begin{proof}
We construct an acyclic graph \( G' = (V', E') \) by creating \( T+1 \) copies of the input graph \( G = (V, E) \) and arranging them in layers, with edges going from each layer \( i \) to layer \( i+1 \).
That is, \( V' = [T+1] \times V \), and \( E' = \{ ((i, u), (i+1, v)) \mid t \in [T], (u,v) \in E \} \).

We then run \( \AlgWalk \) (\cref{alg:walk}) on the graph \( G' \), using start vertex \( (0,s) \) and sink vertex \( (T,t) \), and keeping the same parameter \( \eps \).
Since a random walk on \( G' \) from \( s \) to a sink is equivalent to a \( T \)-step random walk on \( G \), the proof of \cref{thm:randomDag} implies the algorithm will estimate the probability that a \( T \)-step random walk ends at node \( t \) witin the required additive error bound of \( \eps \).
The same proof also implies our algorithm is catalytic.

The space bounds (\( O(\log nmT/\eps) \) working space and \( \tO( nT \log( m / \eps ) ) \) catalytic space) follow directly from the proof of \cref{thm:randomDag}, since \( G' \) has \( nT \) nodes.
The runtime is the time needed to run \( \AlgWalkOnce \)
\( K = O( mT / \eps ) \)  % KVALUE (<- for searching when K changes) but note #nodes = nT here
times.
Each call to \( \AlgWalkOnce \) takes \( \tO( T ) \) time, since every walk will have \( T \) steps.
So, the total runtime is \( \tO( mT^2 / \eps ) \).
\end{proof}

\subsection{What if we don't eliminate cycles?}\label{sec:stationary}

The proof of our algorithm's correctness relies on the graph being acyclic, and so when we are given a general graph, we are forced to pay a penalty in space and time to convert it to an acyclic one.
It is tempting to try avoiding the penalty by running the algorithm directly on a graph with cycles.

As it turns out, we end up with an algorithm that is not catalytic, but still accurately simulates a random walk, in the sense that it approaches the graph's stationary distribution: see \cref{thm:stationary}.

It would be interesting to try to modify this algorithm to be catalytic without incurring a time or space penalty.
This seems challenging, because it is possible to construct a graph that causes \( \AlgWalkOnce \) to lose information that was stored on the catalytic tape: two different initializations of the registers \( R_v \) lead to the same final register values, and so recovery is impossible.
See \cref{fig:LossyCycles}.

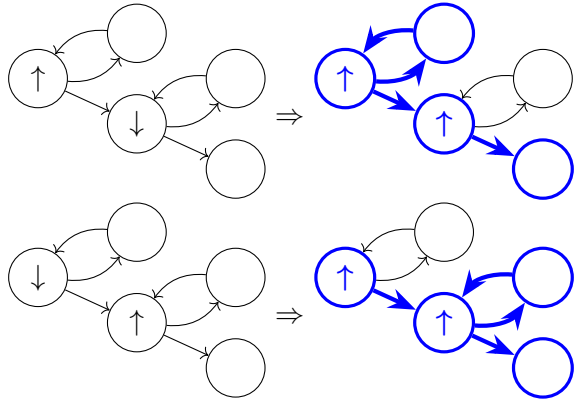
\begin{figure}
\centering
\begin{tikzpicture}[baseline=(11)]
\path
  +(0,0) node (00) [inactive configuration] {\( \uparrow \)}
  +(1.3,0.6) node (10) [inactive configuration] { }
  +(1.3,-0.6) node (11) [inactive configuration] {\( \downarrow \)}
  +(2.6,0) node (20) [inactive configuration] { }
  +(2.6,-1.2) node (21) [inactive configuration] { }
  (00) edge [inactive edge,bend right] (10)
  (10) edge [inactive edge,bend right] (00)
  (00) edge [inactive edge] (11)
  (11) edge [inactive edge,bend right] (20)
  (20) edge [inactive edge,bend right] (11)
  (11) edge [inactive edge] (21);
\end{tikzpicture}
\( \Rightarrow \)
\begin{tikzpicture}[baseline=(11)]
\path
  +(0,0) node (00) [active configuration] {\( \boldsymbol\uparrow \)}
  +(1.3,0.6) node (10) [active configuration] { }
  +(1.3,-0.6) node (11) [active configuration] {\( \boldsymbol\uparrow \)}
  +(2.6,0) node (20) [inactive configuration] { }
  +(2.6,-1.2) node (21) [active configuration] { }
  (00) edge [active edge,bend right] (10)
  (10) edge [active edge,bend right] (00)
  (00) edge [active edge] (11)
  (11) edge [inactive edge,bend right] (20)
  (20) edge [inactive edge,bend right] (11)
  (11) edge [active edge] (21);
\end{tikzpicture}

\begin{tikzpicture}[baseline=(11)]
\path
  +(0,0) node (00) [inactive configuration] {\( \downarrow \)}
  +(1.3,0.6) node (10) [inactive configuration] { }
  +(1.3,-0.6) node (11) [inactive configuration] {\( \uparrow \)}
  +(2.6,0) node (20) [inactive configuration] { }
  +(2.6,-1.2) node (21) [inactive configuration] { }
  (00) edge [inactive edge,bend right] (10)
  (10) edge [inactive edge,bend right] (00)
  (00) edge [inactive edge] (11)
  (11) edge [inactive edge,bend right] (20)
  (20) edge [inactive edge,bend right] (11)
  (11) edge [inactive edge] (21);
\end{tikzpicture}
\( \Rightarrow \)
\begin{tikzpicture}[baseline=(11)]
\path
  +(0,0) node (00) [active configuration] {\( \boldsymbol\uparrow \)}
  +(1.3,0.6) node (10) [inactive configuration] { }
  +(1.3,-0.6) node (11) [active configuration] {\( \boldsymbol\uparrow \)}
  +(2.6,0) node (20) [active configuration] { }
  +(2.6,-1.2) node (21) [active configuration] { }
  (00) edge [inactive edge,bend right] (10)
  (10) edge [inactive edge,bend right] (00)
  (00) edge [active edge] (11)
  (11) edge [active edge,bend right] (20)
  (20) edge [active edge,bend right] (11)
  (11) edge [active edge] (21);
\end{tikzpicture}
\caption{\label{fig:LossyCycles}%
An example showing why the algorithm of \cref{thm:stationary} can't restore its register values, and so is not catalytic. Two runs of the algorithm for \( T' = 4 \) steps are shown. Left: two different ways to initialize two of the catalytic registers, with \( \uparrow, \downarrow \) representing the possible values as in \cref{fig:probsAndCounts}. Right: the resulting walks highlighted in blue, and the final values of those registers. The final register values are the same, so the algorithm cannot know which version on the left to restore the registers to.}
\end{figure}

To describe the behaviour of this new algorithm, we first define four terms:
\begin{definition}[probability vector, stochastic matrix, stationary distribution, mixing time]
    A \emph{probability vector} \( v \in \R^n \) is any vector with nonnegative entries and \( |v|_1 = 1 \).
    A matrix \( W \in \R^{ n \times n } \) is \emph{(left) stochastic} if every column is a probability vector.
    The \emph{stationary distribution} of a stochastic matrix \( W \in \R^{ n \times n } \) is any probability vector \( \pi \) such that \( W \pi = \pi \).
    We say \( W \) \emph{mixes in time \( T \) with error \( \eps \)} if the stationary distribution \( \pi \) is unique and for every probability vector, \( |W^Tv - \pi| \le \eps \).
    (Here, \( W^T \) is the \( T \)-th power of \( W \), not its transpose.)
\end{definition}

We remark that our definition of mixing time implies the Markov chain is ergodic.

\begin{theorem}\label{thm:stationary}
There is an algorithm which, given a graph \( G \) with \( n \) vertices and \( m \) edges together with vertex \( v^* \) and parameters \( T, \delta \in \N \), returns a number \( \rho \) which approximates the stationary probability at \( v^* \) in the following sense.
If the random walk on \( G \) has stationary distribution \( \pi \) and mixes in time \( T \) with accuracy \( \eps \),
then
\[
\left|\rho - \pi(v^*)\right|\le \eps+\delta \ldotp
\]
The algorithm runs in time \( \tO( T m / \delta ) \) and uses space \( \tO( n ) \).
\end{theorem}

\begin{proof}
\newcommand{\CAfter}{ c_{ \mathrm{ after } } }
\newcommand{\CBefore}{ c_{ \mathrm{ before } } }
We modify the algorithm for \cref{thm:randomDag} so that instead of doing \( K \) walks starting at a node \( s \) and ending at a sink, our new algorithm does a single long walk for \( T' = \lceil T ( m + 2 ) / \delta \rceil \) steps, and counts the number of visits to \( v^* \).
We give the algorithm as \cref{alg:stationary}.

For simplicity, we have each register \( R_v \) store a value in \( [\DOut(v)] \) rather than forcing the range of a register to be a power of two as we did in \cref{sec:walkRegisters} for \cref{alg:walkOnce}.
(We could have done that, using the same technique, but it is not useful in this case since our algorithm would still not be catalytic.)

\begin{algorithm}
Registers: one register \( R_v \) in \( [\DOut(v)] \) for every vertex \( v \).\;
Let \( T' = \lceil T ( m + 2 ) / \delta \rceil \)\;
Initialize \( v \) to any node.\;
\( \NVisit \gets 0 \)\;
\For{\( t \gets 1 \) \KwTo \( T' \)}{
	\If{\( v=v^* \)\label{line:StationaryCheck}}{
		\( \NVisit \gets \NVisit + 1 \)\;
	}
	\tcc{Choose an outgoing edge based on \( R_v \), and update \( R_v \).}
	\( r \gets R_v \bmod{\outdeg(v)} \)\;
	\( R_v \gets (R_v + 1) \bmod 2^{ \ell } \)\;
	\( v \gets \outnbr_G(v,r) \)\label{line:StationaryStep}\;
}
\Return{\( \NVisit / T' \)}\;
\caption{\label{alg:stationary}\( \AlgStationary(G,v^*,T,v^*,\delta) \). Parameters: graph $G$, target vertex \( v^* \), mixing time \( T \), accuracy \( \delta \). Returns a number in \( [0,1] \).}
\end{algorithm}

The time and space used by this algorithm are clear enough.
It remains only to prove the final value of \( \NVisit / T' \) is within \( \eps+\delta \) of the stationary probability \( \pi(v*) \).

Let \( W \in \R^{ n \times n } \) be the random walk matrix of \( G \).
The idea behind the argument that follows is that if \( c \in \N^n \) counts the number of visits to each node, then since the algorithm gives equal attention to all of a node's outgoing edges, \( | Wc - c | \) is not too big, from which it follows that \( c \) approximates the stationary distribution.

We begin by introducing notation for counting visits to edges and vertices.
Let \( V \) be the set of vertices of \( G \).

For an edge \( (u, w) \), let \( c(u,w) \) be the number of times the variable \( v \) changes from \( u \) to \( w \) on \cref{line:StationaryStep}; in other words, the number of times the algorithm ``walks'' from \( u \) to \( w \).

Counting visits to vertices raises a subtle distinction.
For vertex \( w \), let \( \CBefore(w) \) be the number of times the variable \( v \) equals \( w \) when the condition at \cref{line:StationaryCheck} is evaluated, and let \( \CAfter(w) \) be the number of times \( v \) equals \( w \) after \cref{line:StationaryStep} is evaluated, so that
\[ \CBefore(u) = \sum_{ w \in V } c(u,w) \quad \mathrm{ and } \quad \CAfter(w) = \sum_{ u \in V } c(u,w) \ldotp \]
We think of both \( \CBefore(w) \) and \( \CAfter(w) \) as ``number of visits'', and they disagree only at the start and end of the walk the algorithm followed.
Specifically, let \( v_0 \) be the value \( v \) was initialized to, and let \( v_{\mathrm{end}} \) be its value at the end of the algorithm.
Then if \( v_0 \not= v_{\mathrm{end}} \), then \( \CAfter(v_0) = \CBefore(v_0)-1 \), \( \CAfter(v_{\mathrm{end}}) = \CBefore(v_{\mathrm{end}})+1 \), and \( \CBefore(w)=\CAfter(w) \) for all other vertices. (If \( v_0=v_{\mathrm{end}} \), then \( \CBefore \) and \( \CAfter \) agree on that vertex too.)

Similar to \cref{alg:walk}, an important property of \cref{alg:stationary} is that for every node \( u \), it uses all outgoing edges from \( u \) almost the same number of times, with the counts differing by at most one: so \( | c(u,w) - \CBefore(u) / \DOut(u) | < 1 \).
Summing over \( u \) then gives
\[ \left| \CAfter(w) - \sum_{ u \in V } \CBefore(u) / \DOut(u) \right| \le \DIn(w) \]
At this point, it becomes useful to think of \( \CBefore, \CAfter \) as vectors in \( \N^n \).
From this point of view, \( \sum_{ u \in V } \CBefore(u) / \DOut(u) \) is just the \( w \)-th coordinate of \( W \CBefore \), and so, summing over all nodes \( w \), we have
\[ | \CAfter - W \CBefore |_1 \le m = \sum_{ w \in V } \DIn(w) \]
Now, \( | \CBefore - \CAfter | \le 2 \), so
\[ | \CBefore - W \CBefore |_1 \le m+2 \ldotp \]
Since \( W \) does not increase \( 1 \)-norms, it follows that
\[ | W^t \CBefore - W^{ t+1 } \CBefore |_1 \le m+2 \]
for every \( t \in \N \), and so
\[ | \CBefore - W^T \CBefore | \le \sum_{ t=0 }^{ T-1 } | W^t \CBefore - W^{ t+1 } \CBefore | \le T ( m + 2 ) \]
and so
\[ \left| \frac{ \CBefore }{ T' } - W^T \frac{ \CBefore }{ T' } \right| \le \frac{ T ( m + 2 ) }{ T' } \le \delta \ldotp \]
Since \( W \) mixes in time \( T \) with accuracy \( \eps \), we have that
\( | W^T \CBefore / T' - \pi | \le \eps \)
and so
\[ \left| \frac{ \CBefore }{ T' } - \pi \right| \le \eps + \delta \ldotp \]
Since the algorithm returns the \( v^* \)-th coordinate of \( \CBefore \), its answer is within additive error \( \eps + \delta \) of \( \pi( v^* ) \).
\end{proof}

\section{Future Directions}

We hope these examples will inspire others to find new efficient catalytic algorithms for these or other problems.
In particular, it would be interesting to avoid the overhead in \cref{thm:randomGeneral} from converting to an acyclic graph --- \cref{thm:stationary} attempts this, but the algorithm fails to be catalytic. For connectivity, our algorithms are all incomparable (in speed, randomness, and revertibility). It would be interesting to obtain a best of both worlds result.

\ifnames
\section*{Acknowledgements}

E.P. thanks Ryan Williams for encouragement to think about algorithms in CL and useful discussions, and Ian Mertz for the suggestion to work over different moduli. J.C. thanks Michal Kouck\'y and the CCC reviewers for useful suggestions. 
\fi

\bibliographystyle{alpha}
\bibliography{bib}

\end{document}